\documentclass[runningheads,a4paper]{llncs}
\usepackage{amssymb}
\usepackage{graphicx}
\usepackage{lineno}
\usepackage{url}
\urldef{\mailsa}\path|{gaertner, abdolahad.noori}@inf.ethz.ch|    
\newcommand{\keywords}[1]{\par\addvspace\baselineskip
\noindent\keywordname\enspace\ignorespaces#1}
\begin{document}
\mainmatter 
\title{Majority Model on Random Regular Graphs}
\titlerunning{Majority Model on Random Regular Graphs}
\author{Bernd G\"artner \and Ahad N. Zehmakan}
\authorrunning{B. G\"artner \and A. N. Zehmakan}
\institute{Department of Computer Science, ETH Zurich,\\
 Switzerland\\
\mailsa\\
\url{}}
\toctitle{Majority model on random regular graphs}
\tocauthor{Bernd~G\"artner and Ahad~N. Zehmakan}
\maketitle
\begin{abstract}
Consider a graph $G=(V,E)$ and an initial random coloring where each vertex $v \in V$ is blue with probability $P_b$ and red otherwise, independently from all other vertices. In each round, all vertices simultaneously switch their color to the most frequent color in their neighborhood and in case of a tie, a vertex keeps its current color. The main goal of the present paper is to analyze the behavior of this basic and natural process on the random $d$-regular graph $\mathbb{G}_{n,d}$. It is shown that for all $\epsilon>0$, $P_b \le 1/2-\epsilon$ results in final complete occupancy by red in $\mathcal{O}(\log_d\log n)$ rounds with high probability, provided that $d\geq c/\epsilon^2$ for a suitable constant $c$. Furthermore, we show that with high probability, $\mathbb{G}_{n,d}$ is immune; i.e., the smallest dynamic monopoly is of linear size. A dynamic monopoly is a subset of vertices that can ``take over'' in the sense that a commonly chosen initial color eventually spreads throughout the whole graph, irrespective of the colors of other vertices. This answers an open question of Peleg~\cite{peleg2014immunity}.
\end{abstract}
\section{Introduction}
\label{Introduction}
Consider a graph $G=(V,E)$ with an initial coloring where each vertex is red or blue. Each red/blue vertex could correspond to an infected/uninfected cell in a brain, a burning/non-burning tree in a forest, a positive/negative individual in a community regarding a reform proposal, or an informed/uninformed processor in a distributed system. Starting from an initial coloring, and in discrete-time rounds, all vertices synchronously update their current color based on a predefined rule as a function of the current coloring of their neighbors. By defining a suitable updating rule, this process can model different basic dynamic phenomena, like infection spreading among cells, fire propagation in a forest, opinion forming regarding an election in a community, or information distribution among processors. As two simple examples, a tree starts burning if at least one of its neighbors is on fire, or a person adopts the most frequent opinion among his/her friends.

Researchers from a wide spectrum of fields, from biology to physics, and with various motivations, have extensively investigated the behavior of such processes. One of the most natural updating rules, whose different variants have attracted a substantial amount of attention, is the \emph{majority rule} where a vertex updates its current color to the most frequent color in its neighborhood.

Here, one of the most studied variants is \emph{majority bootstrap percolation} in which by starting from a random coloring, where each vertex is blue with probability $P_b$ and red with probability $P_r=1-P_b$ independently, in each round a blue vertex switches to red if at least half of its neighbors are red, and a red vertex stays red forever. A considerable amount of effort has been put into the investigation and analysis of majority bootstrap percolation, both theoretically and experimentally, from results by Balogh, Bollobas, and Morris~\cite{balogh2009majority} to the recent paper by Stefansson and Vallier~\cite{stefansson2015majority}. Typical graphs are the $d$-dimensional lattice, the $d$-dimensional hypercube, the binomial random graph, and the random regular graph.

The main motivation behind majority bootstrap percolation is to model monotone dynamic processes like rumor spreading, where an informed individual will always stay informed of the rumor (corresponding to red color, say). However, it does not model non-monotone processes like opinion forming in a community, distributed fault-local mending, and diffusion of two competing technologies over a social network. For this, the following \emph{majority model} is considered: we are given a graph $G=(V,E)$ and an initial random coloring, where each vertex is blue with probability $P_b$ and red otherwise, independently of other vertices. In each round, all vertices simultaneously update their color to the most frequent color in their neighborhood; in case of a tie, a vertex keeps its current color. Since the majority model is a deterministic process on a finite state space, the process must reach a cycle of states after a finite number of rounds. The number of rounds that the process needs to reach the cycle is called the \emph{consensus time} of the process. 

Even though different aspects of the majority model like the consensus time and its threshold behavior have been studied both experimentally and theoretically (see Section \ref{prior works} for more details), there is not much known about the behavior of the majority model on the random $d$-regular graph. Majority bootstrap percolation~\cite{balogh2007bootstrap}, rumor spreading~\cite{Fountoulakis2010,mourrat2016phase}, and flooding process~\cite{amini2013flooding} have been studied on random regular graphs, but this graph class is not easy to handle. Even though the behavior of majority bootstrap percolation on the random regular graph had been discussed in several prior works, it took almost two decades until Balogh and Pittel~\cite{balogh2007bootstrap} could analyze the behavior of the process partially. They proved (under some limitations on the size of $d$) that there are two values $P_1$ and $P_2$ such that $P_r\ll P_1$ results in the coexistence of both colors and $P_2\ll P_r$ results in a fully red configuration with high probability\footnote{We shortly write $f(n)\ll g(n)$ for $f(n)=o(g(n))$. For a graph $G=(V,E)$, we say an event happens with high probability (w.h.p.) if its probability is at least $1-o(1)$ as a function of $|V|$.}; however, $P_1\ne P_2$ which leaves a gap in the desired threshold behavior of the process. 

In the present paper, we prove that in the majority model on the random $d$-regular graph, and for any constant $\epsilon>0$, $P_b \leq 1/2-\epsilon$ results in final complete occupancy by red color in $\mathcal{O}(\log_d\log n)$ rounds w.h.p. if $d\geq c/\epsilon^2$ for a suitable constant $c$. In words, even a narrow majority takes over the whole graph extremely fast. We should point out that the result probably holds for $d\geq 3$, but our proof techniques do not yield this, since they require sufficient edge density in the underlying graph. We also show that the upper bound of $\mathcal{O}(\log_d\log n)$ is best possible. 

A natural context of this result is the \emph{density classification problem}; coming from the theory of cellular automata, this is the problem of finding an updating rule for a given graph $G$ such that for any initial 2-coloring the process reaches a monochromatic configuration by the initial majority color. It turned out that the problem is hard in the sense that even for a cycle, there is no rule which can do the density classification task perfectly~\cite{land1995no}. Our result shows that the majority rule does the density classification task acceptably for almost every $d$-regular graph with $d$ sufficiently large (see Theorem \ref{theorem 5} for the precise meaning of acceptably and sufficiently large).  

It is an interesting (and currently unanswered question) which properties of a graph are chiefly responsible for the majority rule being able to almost classify density - or failing to do so. For example, we know that on a torus $T_{\sqrt{n},\sqrt{n}}$ (a $\sqrt{n} \times \sqrt{n}$ lattice with ``wrap-around''), already a very small initial blue density of $P_b\gg 1/n^{1/4}$ prevents red color from taking over, w.h.p.~\cite{gartner2017color}. A plausible explanation is that the torus $T_{\sqrt{n},\sqrt{n}}$ has a very low vertex/edge expansion in comparison to the random regular graph; however, we do not know whether expansion is indeed the right parameter to look at here. 


As a concrete application of our main Theorem \ref{theorem 5}, we improve a result and answer an open question by Peleg~\cite{peleg2014immunity}. Motivated by the problem of fault-local mending in distributed systems, he introduced the concept of immunity. An $n$-vertex graph $G$ is $(\alpha,\beta)$-\emph{immune} if a set of $m\leq \beta n$ vertices with a common color can take over at most $\alpha m$ vertices in the next round in the majority model. Peleg proved that there exists a $d$-regular graph that is $(\frac{c_2\log n}{d},\beta)$-immune, for suitable constants $c_1,c_2,\beta>0$ and $d\geq c_1$. He also showed that this result is tight up to a logarithmic factor. We close this logarithmic gap. Peleg also asked whether there exist regular graphs that are immune in the sense that no sub-linear size set of a common color can eventually take over the whole graph. We answer his question positively: w.h.p. the random $d$-regular graph is immune.

The outline of the paper is as follows. After presenting basic definitions and prior research in Sections \ref{notation} and \ref{prior works}, the behavior of the majority model on the random $d$-regular graph is analyzed in Section \ref{majority}; the application to immunity is presented in Section \ref{dynamo}.                
\subsection{Notation and Preliminaries}
\label{notation}
For a vertex $v$ in graph $G=(V,E)$ the \emph{neighborhood} of $v$ is defined as $N(v) := \{u\in V: (v,u)\in E\}$. Furthermore for $u,v \in V$, let $d(u,v)$ denote the length of the shortest path between $v,u$ in terms of the number of edges, which is called the \emph{distance} between $v$ and $u$ (for a vertex $v$, we define $d(v,v)=0$). For $v\in V$, $N_i(v):=\{u \in V:d(v,u)\leq i\}$ is the set of vertices in distance at most $i$ from $v$.

 
A \emph{generation} is a function $g:V\rightarrow\{b,r\}$ where $b$ and $r$ stand for blue and red, respectively. In addition to $g(v)=c$ for a vertex $v\in V$ and $c\in \{b,r\}$, we also write $g|_S=c$ for a set $S\subseteq V$ which means $\forall v\in S$, $g(v)=c$. For a graph $G=(V,E)$ and a random initial generation $g_0$, where $\forall v\in V$ $Pr[g_0(v)=b]=P_b$ and $Pr[g_0(v)=r]=P_r=1-P_b$ independently, assume $\forall i\geq 1$ and $v\in V$, $g_i(v)$ is equal to the color that occurs most frequently in $v$'s neighborhood in $g_{i-1}$, and in case of a tie $g_{i}(v)=g_{i-1}(v)$. This model is called the \emph{majority model}. Without loss of generality, we always assume that $P_b \leq P_r$. 

The \emph{random $d$-regular graph} $\mathbb{G}_{n,d}$ is the random graph with a uniform distribution over all $d$-regular graphs on $n$ vertices, say $[n]$ (in this paper, we assume whenever talking about $\mathbb{G}_{n,d}$, $dn$ is even). The definition of the random regular graph is conceptually simple, but it is not easy to use. However, there is an efficient way to generate $\mathbb{G}_{n,d}$ which is called the \emph{configuration model}~\cite{bender1978asymptotic}. 

In the configuration model for $V=[n]$, which is to be the vertex set of the graph, we associate the $d$-element set $W_i=\{i\}\times[d]=\{(i,i'): 1 \leq i' \leq d\}$ to vertex $1\leq i\leq n$. Let $W=[n]\times[d]$ be the union of $W_i$s; then a \emph{configuration} is a partition of $W$ into $dn/2$ pairs. These pairs are called the edges of the configuration. The natural projection of the set $W$ onto $V=[n]$ (ignoring the second coordinate) projects each configuration $F$ to a multigraph $\pi(F)$ on $V$. Note that $\pi(F)$ might contain loops and multiple edges. Thus, $\pi(F)$ is not necessarily a simple graph. We define the \emph{random $d$-regular multigraph} $\mathbb{G}^{\ast}_{n,d}$ to be the multigraph $\pi (F)$ obtained from a configuration $F$ chosen uniformly at random among all configurations on $W$. Bender and Canfield~\cite{bender1978asymptotic} proved that if we consider $\mathbb{G}_{n,d}^{*}$ and condition on it being a simple graph, we obtain a random $d$-regular graph on $V$ with uniform distribution over all such graphs. Furthermore, it is known~\cite{janson2011random} that if $Pr(\mathbb{G}^{*}_{n,d}\in A_n) \rightarrow 0$ as $n \rightarrow \infty$ then also $Pr(\mathbb{G}_{n,d} \in A_n) \rightarrow 0$, where $A_n$ is a subset of $d$-regular multigraphs on $V$. This allows us to work with $\mathbb{G}^{*}_{n,d}$ instead of $\mathbb{G}_{n,d}$ itself in our context. To generate a random configuration, it suffices to define an arbitrary ordering on the elements of $W$ and repeatedly match the first unmatched element in this order with another unmatched element uniformly at random. In Lemma \ref{lemma 1}, we utilize a slightly different construction from~\cite{bollobas1982diameter}.

\subsection{Prior Work}
\label{prior works}
Even though a substantial amount of effort has been put into the study of a wide spectrum of the majority-based dynamic processes, our attention here is mostly devoted to the prior work concerning the majority model. However, let us briefly point out a couple of remarkable accomplishments regarding the majority bootstrap percolation, which is arguably the closest model to ours. Aizenmann and Lebowitz~\cite{aizenman1988metastability} proved that in the $d$-dimensional lattice there is a threshold value $P_c$ so that $P_r\ll P_c$ and $P_c\ll P_r$ respectively result in the stable coexistence of both colors and fully red configuration w.h.p. Balogh and Bollobas~\cite{balogh2009majority} investigated the model on the $d$-dimensional hypercube and proved that the process has a phase transition with a sharp threshold. As discussed, the case of the random regular graph was also studied by Balogh and Pittel~\cite{balogh2007bootstrap}. 

The majority model was introduced by Spitzer~\cite{spitzer1970interaction} in 1970. Afterwards, the model's behavior was investigated mostly by computer simulations (i.e., Monte-Carlo methods). These computer simulations (see for instance~\cite{de1992isotropic}) suggested that the model shows a threshold behavior on the two-dimensional torus $T_{\sqrt{n},\sqrt{n}}$. To address this observation, it was proven~\cite{gartner2017color} that $P_b\ll n^{-1/4}$ and $P_b \gg n^{-1/4}$ respectively result in red monochromatic generation and the stable coexistence of both colors w.h.p. Furthermore Schonmann~\cite{schonmann1990finite} proved in the biased variant of the majority model, where in case of a tie always red is chosen, and torus $T_{\sqrt{n},\sqrt{n}}$, for $1/\log n\ll P_r$ w.h.p. the process reaches fully red generation. 

Since the updating rule is deterministic and there are $2^{|V|}$ possible colorings, the majority process must always reach a cycle of generations. Poljak and Turzik~\cite{poljak1986pre} showed that the number of rounds needed to reach the cycle (i.e., the consensus time) is $\mathcal{O}(|V|^2)$, and Goles and Olivos~\cite{goles1981comportement} proved the length of the cycle is always one or two. Frischknecht, Keller, and Wattenhofer~\cite{frischknecht2013convergence} showed there exists graph $G=(V,E)$ which needs $\Omega(|V|^2/\log^2 |V|)$ rounds to stabilize for some initial coloring in the majority model, which thus leaves only a poly-logarithmic gap. Kasser et al.~\cite{kaaser2015voting} studied a decision variant of the problem; they proved for a given graph $G=(V,E)$ and an integer $k$, it is NP-complete to decide whether there exists an initial coloring for which the consensus time is at least $k$.

Kempe, Kleinberg, and Tardos~\cite{kempe2003maximizing}, motivated from viral marketing, and independently Peleg~\cite{peleg1997local}, motivated from fault-local mending in distributed systems, introduced the concept of dynamic monopoly, a subset of vertices that can take over the whole graph. Afterwards, lots of studies regarding the size of dynamic monopolies and their behavior have been done. To name a few, even though it was conjectured~\cite{peleg1997local} that the size of the smallest dynamic monopoly in the majority model is $\Omega(\sqrt{|V|})$ for a graph $G=(V,E)$, Berger~\cite{berger2001dynamic}, surprisingly, proved there exist graphs with dynamic monopolies of constant size. Furthermore, Flocchini et al.~\cite{flocchini2004dynamic} studied the size of the smallest dynamic monopoly in the two dimensional torus. For more related results regarding dynamic monopolies, the interested reader is referred to a more recent work by Peleg~\cite{peleg2014immunity}. 

\section{Majority Model on Random Regular Graphs} 
\label{majority} 
The three special cases of $d=0,1,2$ are exceptions to many properties of the random $d$-regular graph $\mathbb{G}_{n,d}$. For instance, $\mathbb{G}_{n,d}$ is $d$-connected for $d \geq 3$, but disconnected for $d \leq 2$ w.h.p.~\cite{janson2011random}. In the majority model also these three special cases show a different sort of behavior, which intuitively comes from their disconnectivity. We shortly discuss these cases following two purposes. Firstly, their threshold behavior sounds interesting by its own sake. Secondly, as a warm-up it probably helps the reader to have a better understanding of the majority model before going through our main results and proof techniques concerning the density classification in Section \ref{density classifier} and dynamic monopolies in Section \ref{dynamo}. 

A $0$-regular graph is an empty graph with $n$ vertices, and a $1$-regular graph is the same as a perfect matching. We argue that in both cases $P_b \ll 1/n$ results in red monochromatic generation and $P_b \gg 1/n$ results in the coexistence of both colors w.h.p. (recall we assume $P_b\leq P_r$). Let random variable $X$ denote the number of blue vertices in the initial generation. $\mathbb{E}[X]=nP_b=o(1)$ for $P_b \ll 1/n$, and by Markov's inequality~\cite{feller1968introduction} w.h.p. $g_0|_V=r$. If $1/n \ll P_b$, then $\mathbb{E}[X]=\omega(1)$. Since X is the sum of $n$ independent Bernoulli random variables, Chernoff bound~\cite{feller1968introduction} implies that w.h.p. there exists a blue vertex in the initial generation, which guarantees the survival of blue color in both cases. 

We show the random $2$-regular graph $\mathbb{G}_{n,2}$ also has a phase transition, but at $1/\sqrt{n}$ instead of $1/n$. Actually more strongly, we prove that for any $n$-vertex $2$-regular graph, $P_b\ll 1/\sqrt{n}$ and $1/\sqrt{n}\ll P_b$ w.h.p. result in fully red generation and the stable coexistence of both colors, respectively. Notice a $2$-regular graph is the union of cycles of length at least $3$.

\begin{theorem}
In the majority model and an $n$-vertex $2$-regular graph $G=(V,E)$, $P_b\ll 1/\sqrt{n}$ results in red monochromatic generation and $1/\sqrt{n}\ll P_b$ outputs the stable coexistence of both colors w.h.p.
\end{theorem}
\begin{proof}
In a generation $g$, define a blue (red) edge to be an edge whose both endpoints are blue (red). Consider an arbitrary edge set $E'\subset E$ which contains linearly many disjoint edges (a maximum matching, say), and let random variable $X_1$ denote the number of blue edges of $E'$ in $g_0$. For $1/\sqrt{n}\ll P_b$, $\mathbb{E}[X_1]=\omega(1)$; thus, by Chernoff bound there is a blue edge in $g_0$ w.h.p. which guarantees the survival of blue color.

If $P_b \ll 1/\sqrt{n}$, then $\mathbb{E}[X_2]=o(1)$, where the random variable $X_2$ denotes the number of blue edges in $g_0$, which by Markov's inequality implies there is no blue edge in $g_0$ with high probability. If in a cycle there is no blue edge and there is at least a red edge, then the cycle gets red monochromatic after at most $n/2$ rounds because the red edge grows from both sides in each round until it covers the whole cycle. Thus, it only remains to show that each cycle contains a red edge with high probability. An odd cycle always contains a monochromatic edge (red in our case). Then, let $X_3$ denote the number of even cycles which contain no monochromatic edge; i.e., the vertices are red and blue one by one. Define $n_i$ to be the number of cycles of length $i$. We have

\[
\mathbb{E}[X_3]\leq \sum_{2 \leq i \leq \lfloor n/2\rfloor} n_{2i}\cdot 2P_b^{i}(1-P_b)^{i}\leq 2P_b^{2}\sum_{2 \leq i \leq \lfloor n/2 \rfloor} n_{2i} = o(1)
\]
where we used $P_b\ll 1/\sqrt{n}$ and the fact that there are at most linearly many cycles. Therefore, w.h.p. there is no cycle without a red edge. \qed 
\end{proof}
\subsection{Density Classification}
\label{density classifier}
In this section (Theorem \ref{theorem 5}), it is shown that in the $d$-regular random graph $\mathbb{G}_{n,d}$ and the majority model, $P_b\leq 1/2-\epsilon$, for an arbitrarily small constant $\epsilon>0$, results in fully red generation in $\mathcal{O}(\log_d\log n)$ rounds w.h.p. provided that $d\geq c/\epsilon^2$ for a suitable constant $c$. To prove that, first we need to provide Lemmas \ref{lemma 1} and \ref{lemma 4} as the ingredients, which are also interesting and important by their own sake. Specifically, the results in Section \ref{dynamo} concerning dynamic monopolies and immunity are built on Lemma \ref{lemma 4}.

Let say the $k$-neighborhood of a vertex $v$ in a graph $G$ is a tree if the induced subgraph by vertex set $N_k(v)$ is a tree. Roughly speaking, Lemma \ref{lemma 4} explains that for small $d$ and $k$ the expected number of vertices whose $k$-neighborhood is not a tree in $\mathbb{G}_{n,d}$ is small. This local tree-like structure turns out to be very useful in bounding the consensus time of the process.
\begin{lemma}
\label{lemma 1}
In $\mathbb{G}_{n,d}$, the expected number of vertices whose $k$-neighborhood is not a tree is at most $4d^{2k}$.
\end{lemma}
\begin{proof}
Firstly, we assume $k< \log_d(n/2)$ because otherwise the statement is clearly true. Furthermore as discussed in Section \ref{notation}, we work with the random $d$-regular multigraph $\mathbb{G}_{n,d}^{*}$ instead of the random $d$-regular graph $\mathbb{G}_{n,d}$, on vertex set $V=[n]$. We generate a uniformly at random configuration by partitioning $W=\bigcup_{1\leq i\leq n}W_i$ into $dn/2$ pairs as follows, where the $d$-element set $W_i=\{i\}\times [d]$ corresponds to the vertex $1\leq i\leq n$. In step 1, we start from an arbitrary class (we utilize the terms of $d$-element set and class interchangeably), say $W_1$, and match its elements one by one based on an arbitrary predefined order with an unmatched element (from $W_1$ or other classes) uniformly at random. We say a class has been \emph{reached} in step $j\geq 1$ if for the first time in step $j$ one of its elements has been matched. In step $j\geq 2$, we match the unmatched elements of the classes reached in step $j-1$ one by one based on a predefined ordering, say lexicographical order, with unmatched elements uniformly at random. It is possible in some step, no new class is reached. In this case, if all elements are matched, the process is over; otherwise we continue the process from one of the unreached classes, say the one with the smallest index.\footnote{For a more formal description of the construction, please see~\cite{bollobas1982diameter}, and notice since the second element always is chosen randomly, the generated configuration is random.}

Let say in step $j\geq 1$ we match element $x$ with an element $y$, chosen uniformly at random among all yet unmatched elements. One says $xy$ is a \emph{cycle-maker} if $y$ is not the first element matched in its class. The probability that an edge selected in the $j$-th step is a cycle-maker is at most $d^{j}/(n-d^j)$. Thus, the probability that there is a cycle-maker edge in the first $k$ steps is at most $2d^{k}\cdot\max_{1\leq j \leq k}\frac{d^j}{n-d^j}$ which is smaller than $\frac{2d^{2k}}{n-d^k}$. Let $X$ denote the number of vertices whose $k$-neighborhood is not a tree. Then, we have $\mathbb{E}[X]\leq (2nd^{2k})/(n-d^k)$ which is smaller than $ 4d^{2k}$ for $k< \log_d (n/2)$ because $n-d^{k}> n-d^{\log_d (n/2)}= n/2$. \qed  
\end{proof}
\begin{corollary}
\label{corollary 3}
In $\mathbb{G}_{n,d}$, the number of vertices whose $(c'\log_d\log_2n)$-neighborhood is not a tree is at most $\log_{2}^{2c'+1}n$ w.h.p., for constant $c'>0$.
\end{corollary}
\begin{proof}
Let $X$ denote the number of vertices whose $(c'\log_d\log_2n)$-neighborhood is not a tree. By Lemma \ref{lemma 1}, $\mathbb{E}[X] \leq 4d^{2c'\log_d\log_2n}=4\log_{2}^{2c'}n$.
By Markov's inequality $Pr[X\geq \log_{2}^{2c'+1}n]\leq 4/\log_2n=o(1)$. \qed
\end{proof}
In a graph $G=(V,E)$ for two (not necessarily disjoint) vertex sets $S$ and $S'$, we say that $S$ \emph{controls} $S'$ if $S$ being monochromatic in some generation implies $S'$ being monochromatic (of the same color) in the next generation in the majority model, irrespective of the colors of other vertices. Clearly in $\mathbb{G}_{n,d}$, $S$ controls $S'$ implies that for every $v\in S'$ at least $\lceil d/2\rceil$ of its neighbors are in $S$. 
\begin{lemma}
\label{lemma 4}
In $\mathbb{G}_{n,d}$ on vertex set $V=[n]$ with $d\geq c_1$, w.h.p. there do not exist two vertex sets $S,S'$ such that $S$ controls $S'$, $|S|\leq \frac{n}{c^{\prime \prime}}$, $|S'|=\lceil \frac{10|S|}{d}\rceil$, where $c_1,c^{\prime \prime}$ are sufficiently large constants.
\end{lemma}
This immediately implies that in $\mathbb{G}_{n,d}$ and the majority model, less than $\frac{n}{c^{\prime\prime}}$ blue (red) vertices will die out in $\mathcal{O}(\log_dn)$ rounds, with high probability.
\begin{proof} 
We fix two sets $S,S'$ of the given sizes $s$ and $s'$. We show that the probability for $S$ controlling $S'$ is so small that a union bound over all pairs $(S,S')$ yields the desired high probability result. We equivalently work in $\mathbb{G}^{\ast}_{n,d}$ the relevant ``initial'' part of which we generate as follows: we iterate through the pairs in $S'\times[d]$ in some fixed order and match each yet unmatched pair with a random unmatched pair in $V\times[d]$. In order for $S$ to control $S'$, at least $\lceil d/2\rceil$ of the $d$ pairs $(v,i)$ must get matched with pairs in $S\times[d]$, for every $v\in S'$. Overall, at least $\lceil d/2\rceil s'$ of the $ds'$ pairs $S'\times[d]$ get matched with pairs in $S\times[d]$. Such a match is established only when the randomly chosen partner happens to be in $(S\cup S')\times[d]$, and this may actually yield two of the required $\lceil d/2\rceil s'$ pairs. Hence, for $S$ to control $S'$, at least $\ell:=\lceil d/2\rceil s'/2$ of the $L:=ds'$ iterations must be \emph{active}, meaning that they match a yet unmatched pair with a pair in $(S\cup S')\times[d]$. \newcommand{\bb}{\mathbf{b}} For a bit vector $\bb$ of length at most $L$, let $A(\bb)$ denote the event that iteration $i$ is active for exactly the indices where $b_i=1$. Then $Pr[A(b_1,b_2\ldots,b_L)] = \prod_{i=1}^LPr[\mbox{iteration $i$ is active if $b_i=1$}|A(b_1,\ldots,b_{i-1})]$ (the right-hand side is a telescoping product). Now, irrespective of $b_1,\ldots,b_{i-1}$, an iteration is active with probability at most $d(s+s')/(nd-2ds') = (s+s')/(n-2s')\leq 2s/n$. Hence, for a vector $\bb$ with at least $\ell$ ones, $Pr[A(b_1,b_2\ldots,b_L)]\leq (2s/n)^{\ell}$. As there are at most $2^L$ such vectors, the probability that at least $\ell$ iterations are active is at most $2^L(2s/n)^{\ell}\leq 2^{10s}(2s/n)^{\frac{5}{2}s}$. Hence by a union bound, the probability $P$ that there exist such sets $S$ and $S'$ in a random configuration is at most $\sum_{s=1}^{\frac{n}{c^{\prime\prime}}} {n \choose s} {n \choose \lceil \frac{10s}{d}\rceil} 2^{10s}(\frac{2s}{n})^{\frac{5}{2}s}$. Since $d\geq c_1$ for a large constant $c_1$, ${n \choose \lceil \frac{10s}{d}\rceil}\leq{n \choose s}$; thus, applying Stirling's approximation~\cite{feller1968introduction} (i.e., ${n \choose k}\leq (ne/k)^k$) yields $P\leq \sum_{s=1}^{\frac{n}{c^{\prime\prime}}} (\frac{ne}{s})^{2s} 2^{10s}(\frac{2s}{n})^{\frac{5}{2}s}$. Furthermore, since $e^{2s}\cdot2^{10s}\cdot2^{\frac{5}{2}s}\leq (c^{\prime\prime})^{s/4}\leq (n/s)^{s/4}$ for sufficiently large $c^{\prime\prime}$, we have $P \leq \sum_{s=1}^{\frac{n}{c^{\prime\prime}}} (\frac{n}{s})^{2s} (\frac{s}{n})^{\frac{9}{4}s}=\sum_{s=1}^{\frac{n}{c^{\prime\prime}}}(\frac{s}{n})^{\frac{s}{4}}=o(1)$. \qed
\end{proof}
As will be discussed in the proof of Theorem \ref{theorem 5}, for the majority model on $\mathbb{G}_{n,d}$ with $P_b= 1/2-\epsilon$ and $d\geq c/\epsilon^2$ there is a simple argument which shows that the expected density of the blue vertices drops from $1/2-\epsilon$ in $g_0$ to an arbitrarily small constant in $g_1$ if we select the constant $c$ sufficiently large. Therefore, one might want to apply Lemma \ref{lemma 4} to show the process w.h.p. gets red monochromatic in $\mathcal{O}(\log_d n)$ rounds. However, in Theorem \ref{theorem 5} we show actually $\mathcal{O}(\log_d\log n)$ rounds suffice to get red monochromatic with high probability. To prove that, we need the tree structure argued in Lemma \ref{lemma 1}.
 
\begin{theorem}
\label{theorem 5}
In the majority model and $\mathbb{G}_{n,d}$, by starting from $P_b\leq 1/2-\epsilon$, for an arbitrarily small constant $\epsilon>0$, the process gets red monochromatic in $\mathcal{O}(\log_d\log n)$ rounds w.h.p. provided that $d\geq c/\epsilon^2$ for suitable constant $c$. 
\end{theorem}
\begin{proof}
Let say a vertex is in the $j$-th level of a rooted tree if its distance to the root is $j$. Now, consider a tree $T$ rooted at vertex $v$ and of height $k$ so that except the vertices in the $k$-th level (i.e., leaves), all vertices are of degree $d$. We consider the following process on $T$, which we call the \textit{propagation process}, where in the initial configuration all the internal vertices are \emph{inactive} and each leaf is blue with probability $P_b$ and red with probability $1-P_b$ independently. Assume in each round an inactive vertex whose all children are colored adopts color blue if at least $\lfloor (d-1)/2 \rfloor$ of its children are blue and red otherwise. Clearly after $k$ rounds, the root (vertex $v$) is colored with blue or red. Let $P_i$ for $0 \leq i \leq k$ denote the probability that a vertex in the $(k-i)$-th level is blue after round $i$; specifically, $P_k$ is the probability that vertex $v$ is blue at the end of the process. More accurately, $P_0=P_b$ and for $1 \leq i \leq k$ $P_i=\sum_{j=\lfloor (d-1)/2\rfloor}^{d-1}{d-1 \choose j} P_{i-1}^{j}(1-P_{i-1})^{d-1-j}$.

Now, let us get back to the majority model and the random $d$-regular graph $\mathbb{G}_{n,d}$. Consider a vertex $v$ so that the induced subgraph by $N_k(v)$ is a tree $T$. Clearly in $T$, except the vertices in the $k$-th level, all vertices are of degree $d$. Now, we claim the probability that vertex $v$ is blue in generation $g_k$ in the majority model is at most $P_k$, which is equivalent to the probability that the root of $T$ is blue in the propagation process after $k$ rounds, with the same $P_b$. This is true because by starting with the same coloring for the leaves of $T$ (the vertices in distance $k$ from root $v$) if in the $k$-th round in the propagation process the root is red, it is also red in the majority model and generation $g_k$, irrespective of the colors of other vertices. For $k=1$ this is trivially true. Now, do the induction on $k$; if the root is red in the propagation process after $k$-th round, it means that there exist less than $\lfloor (d-1)/2 \rfloor$ blue vertices among root's children in round $k-1$ which implies by the induction hypothesis there are less than $\lfloor (d-1)/2 \rfloor$ blue vertices in $v$'s neighborhood in $g_{k-1}$ in the majority model; then $g_k(v)=r$.

So far, we showed the probability of being blue in $g_k$ for a vertex, whose $k$-neighborhood is a tree, is at most $P_k$ with $P_0=P_b$. Now, we upper-bound the probability $P_k$. Without loss of generality, assume $d$ is odd, and suppose $d'=d-1$. First, let us bound $P_1$; clearly, $P_1\leq \sum_{j=d'/2}^{d'}{d' \choose j} (1/2-\epsilon)^{j}(1/2+\epsilon)^{d'-j}$ which is smaller than
\[ (1/2-\epsilon)^{d'/2}(1/2+\epsilon)^{d'/2}\sum_{j=d'/2}^{d'}{d' \choose j}\leq (1/4-\epsilon^2)^{d'/2}2^{d'}= (1-4\epsilon^2)^{d'/2}.
\]
By applying the estimate $1-x\leq e^{-x}$, we have $P_1\leq e^{-2d'\epsilon^2}$. For $d\geq c_1'\log n$, where $c_1'$ is a large constant, clearly $P_1\leq 1/n^2$ which implies the expected number of blue vertices in $g_1$ is at most $1/n$; i.e., the process gets red monochromatic in one round with high probability. Thus, it only remains to discuss the case of $d\leq c_1'\log n$ for an arbitrarily large constant $c_1'$; in this case, since $P_1\leq e^{-2d'\epsilon^2}$, selecting suitble constant $c$, for $d\geq c/\epsilon^2$, results in $P_1\leq 1/16$. Now, we show $P_i\leq P_{i-1}^{d'/4}$ for $P_{i-1}\leq 1/16$, which yields $P_k\leq 1/n^2$ for $k=c'\log_d\log_2 n$ by selecting constant $c'$ large enough. We know $P_i\leq P_{i-1}^{d'/2}\sum_{j=d'/2}^{d'}{d' \choose j}\leq P_{i-1}^{d'/2}2^{d'}$. Thus, by utilizing $P_{i-1}\leq 1/16$, one has $P_i\leq P_{i-1}^{d'/4}$. Now, let random variable $X_1$ ($X_2$) denote the number of vertices whose $k$-neighborhood for $k=c'\log_d\log_2 n$, is (not) a tree and are blue in $g_k$. We know $\mathbb{E}[X_1]\leq n P_k\leq 1/n$, which implies $X_1=0$ w.h.p. by Markov's inequality. Furthermore, by using Corollary \ref{corollary 3} with high probability $X_2 \leq \log_{2}^{2c'+1}n$. Hence, with high probability the number of blue vertices in $g_k$ is upper bounded by $\log_{2}^{2c'+1}n$. However based on Lemma \ref{lemma 4}, poly-logarithmically many blue vertices die out in $\mathcal{O}(\log_d\log n)$ rounds w.h.p. which finishes the proof. \qed 
\end{proof}
Now, we argue that the bound of $\mathcal{O}(\log_d\log n)$ is tight. We prove in $\mathbb{G}_{n,d}$ and for a constant small initial density $P_b$, say $P_b=1/4$, after $k'=\frac{\log_d\log_2n}{2}$ rounds w.h.p. there exist some blue vertices. We claim in $\mathbb{G}_{n,d}$ there are $\sqrt{n}$ vertices, say $u_1,\cdots,u_{\sqrt{n}}$, whose $k'$-neighborhood is pairwise disjoint. Define indicator random variable $x_i$ to be 1 if $g_0|_{N_{k'}(u_i)}=b$. Clearly, $Pr[x_i=1]\geq(1/4)^{2d^{k'}}=1/2^{4\sqrt{\log_2n}}$. Let $X=\sum_{i=1}^{\sqrt{n}}x_i$; then $\mathbb{E}[X]\geq \sqrt{n}/2^{4\sqrt{\log_2n}}=\omega(1)$. By using Chernoff bound, there exists a vertex $v$ so that $g_0|_{N_{k'}(v)}=b$, which implies $g_{k'}(v)=b$. Now, we prove that there exist $\sqrt{n}$ vertices whose $k'$-neighborhood is pairwise disjoint in every $d$-regular graph. For a $d$-regular graph $G$ by starting from the state that all vertices are \emph{unmarked}, recursively we choose an arbitrary unmarked vertex $u$ and add $u$ to set $U$, which is initially empty, and mark all vertices in $N_{2k'}(u)$. Clearly, the vertices in set $U$ have our required disjointness property, and set $U$ will be of size larger than $\sqrt{n}$ at the end because in each step we mark at most poly-logarithmically many vertices while we start with linearly many unmarked vertices. 
\subsection{Dynamic Monopoly and Immunity}
\label{dynamo}
In distributed systems, the resolution of inconsistencies by the majority rule is a common tool; the idea is to keep redundant copies of data and perform the majority rule to overcome the damage caused by failures. Motivated from this application, one might be interested in the networks in which no small subset of malicious/failed processors can take over a large part of the network. To address this issue, Peleg~\cite{peleg2014immunity} suggested the concept of immunity. An $n$-vertex graph $G$ is $(\alpha,\beta)$-\emph{immune} if a set of $m\leq \beta n$ vertices with a common color can take over at most $\alpha m$ vertices in the next round. One is interested in graphs which are $(\alpha,\beta)$-immune for constant $\beta$ and small $\alpha$ because roughly speaking these graphs are acceptably tolerant of malicious/failed vertices (processors). Peleg~\cite{peleg2014immunity} proved the following theorem regarding the existence of such graphs.
\begin{theorem}\cite{peleg2014immunity}
\label{theorem 6}
There exist constants $c_1,c_2,\beta>0$ such that for every $d\geq c_1$, there exists a $d$-regular graph $G$ which is $(\frac{c_2\log n}{d},\beta)$-immune.
\end{theorem}
Peleg also argued that this result is tight up to a logarithmic factor, meaning there is a constant $c_2>0$ such that for any constant $\beta>0$, there exist no $(\frac{c_2}{d},\beta)$-immune $d$-regular graph. Now as an immediate implication of Lemma \ref{lemma 4}, we present Corollary \ref{corollary 1} which improves upon Peleg's results by removing the extra logarithmic term. Hence, this result is tight up to a constant.
\begin{corollary}
\label{corollary 1}
There exist constants $c_1,c_2,\beta>0$ such that for every $d\geq c_1$, there exists a $d$-regular graph $G$ which is $(\frac{c_2}{d},\beta)$-immune.
\end{corollary}
Furthermore, we say a graph is \emph{immune} if the smallest dynamic monopoly is of linear size, in terms of the number of vertices. We recall that for a graph $G=(V,E)$, a set $D \subseteq V$ is called a dynamic monopoly whenever the following holds: if in the initial generation all vertices of $D$ are blue (red) then the process reaches the blue (red) monochromatic generation, irrespective of the colors of other vertices. As an open problem, Peleg~\cite{peleg2014immunity} asked that whether there exist regular immune graphs. The existence of immune $d$-regular graphs for $d\gg \log n$ is straightforward from Theorem \ref{theorem 6}, but, the question is unanswered for small $d$, while one is more interested in sparse immune regular graphs from a practical, or even a theoretical, perspective. Again, as an immediate result of Lemma \ref{lemma 4}, we have Corollary \ref{corollary 2} which actually represents a stronger statement. 
\begin{corollary}
\label{corollary 2}
$\mathbb{G}_{n,d}$ with $d\geq c_1$ is immune w.h.p. for large constant $c_1$.
\end{corollary} 
\section*{Conclusion} We claim our techniques can be applied to analyze the behavior of the majority model on the binomial random graph $\mathbb{G}_{n,p}$. For $p\gg \log n/n$, one can show $P_b\leq 1/2-\epsilon$ results in fully blue generation in one round w.h.p. by using the argument regarding the case of $d\geq c_1'\log n$ in the proof of Theorem \ref{theorem 5}. For $p\ll \log n/n$ the graph contains a red and a blue isolated vertex in $g_0$ w.h.p. for a fixed $P_b>0$ which result in the coexistence of both colors. 

Exploring the relation between the behavior of the majority model and the expansion level of the underlying graph can be a prospective research direction. Specifically, it would be interesting to prove that graphs with some certain level of expansion have a density classification behavior similar to $\mathbb{G}_{n,d}$.
\subsubsection*{Acknowledgments.} The authors would like to thank Mohsen Ghaffari for several stimulating conversations and Jozsef Balogh and Nick Wormald for referring to some relevant prior results.
\bibliographystyle{acm} %
\bibliography{RandomRegularMajority}
\end{document}